\documentclass[
nolineno, authorcolumns,a4paper,english,cleveref,autoref,thm-restate,nameinlink,colorlinks=true,linkcolor=NavyBlue,citecolor=NavyBlue,urlcolor=NavyBlue]{gd-lipics-v3}

\hideLIPIcs
\usepackage{amssymb,amsthm,amsmath}
\usepackage{thmtools}
\usepackage[dvipsnames]{xcolor}
\usepackage{xspace}
\usepackage{float}
\usepackage[linesnumbered, lined, boxed]{algorithm2e}

\usepackage{todonotes}
\usepackage{complexity}
\usepackage{apptools}

\newcommand{\restateref}[1]{\IfAppendix{\hyperref[#1]{$\star$}}{\hyperref[#1*]{$\star$}}}

\Crefname{ourclaim}{Claim}{Claims}

\let\emph\relax\DeclareTextFontCommand{\emph}{\color{NavyBlue}\em}
\SetArgSty{textrm}

\newcommand{\bpn}{\operatorname{bp}}
\newcommand{\rv}{\operatorname{rv}}

\newcommand{\floor}[1]{{\left\lfloor #1 \right\rfloor}}

\newcommand{\Oh}[1]{O\left(#1\right)}

\usepackage{standalone}
\title{How Close is a Tree to a Euclidean Minimum Spanning Tree?}

\author{Todor Anti\'{c}}{Department of Applied Mathematics, Faculty of Mathematics and Physics, Charles University, Prague, Czech Republic}{todor@kam.mff.cuni.cz}{orcid}{Supported by grant no. 23-04949X of the Czech Science Foundation (GA\v{C}R) and grant SVV–2025–260822 of GA UK.}

\author{Jiří Fiala}{Department of Applied Mathematics, Faculty of Mathematics and Physics, Charles University, Prague, Czech Republic}{fiala@kam.mff.cuni.cz}{https://orcid.org/0000-0002-8108-567X}{Supported by grant no. 25-16847S of the Czech Science Foundation (GAČR).}

\author{Jelena Gliši\'{c}}{Department of Applied Mathematics, Faculty of Mathematics and Physics, Charles University, Prague, Czech Republic}{glisic@kam.mff.cuni.cz}{https://orcid.org/0009-0002-0792-3070}{Supported by grant no. 23-04949X of the Czech Science Foundation (GA\v{C}R) and grant SVV–2025–260822 of GA UK.}

\author{Grzegorz Gutowski}
{Institute~of~Theoretical~Computer~Science, Faculty~of~Mathematics~and~Computer~Science, Jagiellonian~University, Krak{\'o}w, Poland \and \url{https://grzegorz.gutowscy.pl}}
{grzegorz.gutowski@uj.edu.pl}
{https://orcid.org/0000-0003-3313-1237}
{Partially supported by the National Science Centre, Poland,\\ grant no.~2023/51/B/ST6/02833.}

\author{\texorpdfstring{Konstanty\\ Junosza-Szaniawski}{Konstanty Junosza-Szaniawski}}
{Faculty of Mathematics and Information Science, Warsaw University of Technology, Warsaw, Poland}
{konstanty.szaniawski@pw.edu.pl}
{https://orcid.org/0000-0003-0352-8583}
{}

\author{Jan Kratochvíl}{Department of Applied Mathematics, Faculty of Mathematics and Physics, Charles University, Prague, Czech Republic}{honza@kam.mff.cuni.cz}{https://orcid.org/0000-0002-2620-6133}{Supported by grant no. 23-04949X of the Czech Science Foundation (GA\v{C}R).}

\author{Giuseppe Liotta}{Universit\`a degli Studi di Perugia, Perugia, Italy}{giuseppe.liotta@unipg.it}{https://orcid.org/0000-0002-2886-9694}{}

\author{Morteza Saghafian}{Algorithm and Complexity Group, TU Wien, Vienna, Austria}{msaghafi@ac.tuwien.ac.at}{https://orcid.org/0000-0002-4201-5775}{}

\author{Maria Saumell}{Department of Theoretical Computer Science, Faculty of Information Technology, Czech Technical University in Prague, Prague, Czech Republic}{maria.saumell@fit.cvut.cz}{https://orcid.org/0000-0002-4704-2609}{Supported by grant no. 23-04949X of the Czech Science Foundation (GA\v{C}R).}

\author{Krisztina Szilágyi}{Department of Theoretical Computer Science, Faculty of Information Technology, Czech Technical University in Prague, Prague, Czech Republic}{szilakri@fit.cvut.cz}{https://orcid.org/0000-0003-3570-0528}{Supported under the project Robotics and advanced industrial production (reg. no. CZ.02.01.01/00/22\_008/0004590) and CTU Global Postdoc Fellowship Program.}

\author{Pavel Valtr}{Department of Applied Mathematics, Faculty of Mathematics and Physics, Charles University, Prague, Czech Republic}{valtr@kam.mff.cuni.cz}{https://orcid.org/0000-0002-3102-4166}{Supported by grant no. 23-04949X of the Czech Science Foundation (GA\v{C}R).}

\authorrunning{T. Anti\'{c} et al.}
\Copyright{Todor Anti\'{c}, Jiří Fiala, Jelena Gliši\'{c}, Grzegorz Gutowski, Konstanty Junosza-Szaniawski, Jan Kratochvíl, Giuseppe Liotta, Morteza Saghafian, Maria Saumell, Krisztina Szilágyi, and Pavel Valtr}
\ccsdesc[500]{Human-centered computing~Graph drawings}

\acknowledgements{This work was started at the workshop HOMONOLO 2025 in Nová Louka. The authors thank the organizers and all other participants for a productive and inspiring atmosphere.
We also thank anonymous reviewers for their helpful comments.}

\keywords{Straight-line Drawing of Tree, Euclidean Minimum Spanning Tree, 
EMST-Drawability, Caterpillar, Graph Drawing Algorithm}

\category{Track 1, Long Paper}

\begin{document}

\maketitle

\begin{abstract}

Let $\Gamma$ be a straight-line crossing-free drawing of a tree $T$. A \emph{bad pair} in  $\Gamma$ is a pair of non-adjacent vertices of $T$ whose Euclidean  distance in $\Gamma$ is smaller than the length of the longest edge in the path connecting them in~$\Gamma$. When $\Gamma$ has no bad pairs, $\Gamma$ is a Euclidean Minimum Spanning Tree of its vertex set (or EMST-drawing for short). 
Deciding whether a tree of maximum degree at most six
admits an EMST-drawing is known to be \NP-hard. In contrast, we characterize those caterpillars that admit an EMST-drawing. The characterization gives rise to a linear-time algorithm that decides if a caterpillar admits an EMST-drawing, and in the affirmative case, computes such a drawing. For caterpillars of maximum degree six, we further present a linear-time algorithm to compute a crossing-free straight-line drawing with the minimum number of bad pairs. 
For $n$-vertex trees with maximum vertex degree $\Delta$, we prove the $\Delta^2n\log n$ upper bound on the minimum number of bad pairs.
In the special case of stars, we construct a drawing with the minimum number of bad pairs.

\subparagraph{Generative AI Declaration} Generative AI was not used in the preparation of this article.
\end{abstract}

\section{Introduction}\label{se:intro}

A \emph{proximity drawing} $\Gamma$ of a graph $G$ is a straight-line drawing of $G$ in the plane in which pairs of adjacent vertices  are represented ``relatively close'' to each other, while non-adjacent pairs are kept ``relatively far apart''.
Different types of proximity drawings arise from different definitions of ``closeness''.

Minimum weight triangulations and EMST-drawings are notable examples of these proximity drawings.
Specifically, an \emph{EMST-drawing} $\Gamma$ of a tree $T$ is a crossing-free straight-line drawing of $T$ such that $\Gamma$ is  a Euclidean minimum spanning tree of its vertex set.
The  reader is referred to~\cite{DBLP:reference/crc/Liotta13} for definitions of different types of proximity drawings and to~\cite{DBLP:conf/gd/HaaseKLL23,DBLP:conf/gd/HaaseKLL25,DBLP:journals/tcs/LenhartL23} for a limited list of recent results. 

Motivated by the observation that only restricted families of graphs admit proximity drawings under either  local or  global proximity rules~\cite{DBLP:reference/crc/Liotta13}, we initiate the study of graph drawings for which the number of violations of a given proximity rule is as small as possible. In this paper we focus on drawings of trees under the global proximity rule.

More formally, let $\Gamma$ be a straight-line crossing-free drawing of a tree $T$. A \emph{bad pair} of  $\Gamma$ is a pair of non-adjacent vertices of $T$ whose Euclidean  distance in $\Gamma$ is smaller than the length of the longest edge in the path connecting them in~$\Gamma$. We aim to compute a drawing of $T$ that minimizes the number of bad pairs. Observe that a drawing with no bad pairs is an EMST-drawing. When $T$ admits an EMST-drawing we say that $T$ is \emph{EMST-drawable}.

\subparagraph*{Related Results.}  
Monma and Suri~\cite{monma1992transitions} showed that every tree with maximum degree at most five is EMST-drawable and described a linear-time algorithm (in the real RAM model of computation) to construct an EMST-drawing of such trees. Their algorithm uses a grid of size $O\bigl(2^{n^2}\bigr) \times O\bigl(2^{n^2}\bigr)$, and they conjectured the existence of a tree $T$ with $n$ vertices for which any EMST-drawing requires area at least $c^n \times c^n$ for some constant $c>1$. This long-standing conjecture was later confirmed by Angelini et al.~\cite{DBLP:journals/comgeo/AngeliniBCFKS14} for trees of maximum degree five. In contrast, Frati and Kaufmann~\cite{DBLP:journals/comgeo/FratiK11} showed that trees of maximum degree at most four admit EMST-drawings of polynomial area.
Monma and Suri~\cite{monma1992transitions} also observed that any tree containing a vertex of degree greater than six does not admit an EMST-drawing. For trees of maximum degree six, Eades and Whitesides~\cite{EadesW1996algo} proved that deciding whether such a tree is EMST-drawable is \NP-hard.

\subparagraph*{Our Results.} 
 
The following time complexities assume the real RAM model of computation.
\begin{itemize}
    \item In contrast to the \NP-hardness result of Eades and Whitesides~\cite{EadesW1996algo}, in \cref{se:MST-caterpillars} we  characterize the caterpillars  that are EMST-drawable. 
    Our characterization gives rise to a linear-time algorithm
to decide whether a caterpillar $T$ is 
EMST-drawable, and in the affirmative case, computes an EMST-drawing of $T$.  
    
    \item Any EMST-drawable tree has maximum degree at most six~\cite{monma1992transitions}. For caterpillars of maximum degree six that are not EMST-drawable, we show in \cref{se:caterpillars-bp-minimization} how to compute, in linear time, a straight-line crossing-free drawing with minimum number of bad pairs.
    \item  \cref{se:beyond} studies crossing-free straight-line drawings of trees having maximum vertex degree $\Delta > 6$. We show how to compute drawings of stars with the minimum number of bad pairs and describe a linear-time algorithm that constructs a  drawing of any tree of maximum degree $\Delta$ with the 
    $\Delta^2n\log n$ upper bound on the number of bad pairs.
\end{itemize}

Finally, conclusions and open problems can be found in \cref{se:conclusion}. 

Statements whose full proofs can be found in the appendix are marked by a clickable~($\star$).

\section{Preliminaries}

A straight-line drawing $\Gamma$ of a graph $G=(V,E)$ assigns a distinct point $\Gamma(v)$ to every vertex $v\in V$.
For any two vertices $u,v \in V$, we denote by $d_\Gamma(u,v)$ the Euclidean distance between points $\Gamma(u)$ and $\Gamma(v)$.
We use EMST as an abbreviation for Euclidean Minimum Spanning Tree. A crossing in a straight-line drawing is a point of intersection of two edges that is an interior point of at least one of the edges.

We call \emph{planar drawing} a straight-line drawing without crossings.

\begin{definition}
    Let $\Gamma$ be a planar drawing of a tree $T=(V, E)$. We say that $u,v\in V$ form a \emph{non-edge} if $uv\not\in E$. We say that vertices $u$ and $v$ form a \emph{bad pair} if there is an edge $wz$ on the unique path from $u$ to $v$ in $T$ such that $d_\Gamma(u,v)< d_\Gamma(w,z)$. We call the edge $wz$ a \emph{long edge} for the bad pair $u$ and $v$. 
A pair of vertices that is not a bad pair is a \emph{good pair}.
\end{definition}

\begin{definition}
    The \emph{bad pair number} $\bpn(\Gamma)$ of a planar drawing $\Gamma$ is the number of bad pairs in~$\Gamma$.
    The \emph{bad pair number} $\bpn(T)$ of a tree $T$ is the smallest number of bad pairs in any planar drawing of $T$.
\end{definition}
In particular, if a tree $T$ has $\bpn(T)=0$, then $T$ is EMST-drawable.

\begin{definition}

An angle $\varepsilon_v>0$ around a leaf vertex $v$ with neighbor $u$ in an EMST-drawing of a tree is a \emph{safe angle} if rotating the edge $uv$ around $u$ by an angle $\varphi\le\varepsilon_v$ in either direction results in an EMST-drawing.
A \emph{strong EMST-drawing} of a tree is an EMST-drawing such that, if a vertex of degree at most 5 is adjacent to a leaf, then it has an incident leaf with a safe angle.   
\end{definition}

\begin{observation}[\cite{monma1992transitions}]\label{obs:MSTmax6}
    Any EMST-drawable tree has maximum vertex degree 6. 
\end{observation}

In \cref{alg:monma-suri}, we describe the algorithm of Monma and Suri~\cite{monma1992transitions} for drawing a 5-regular tree (i.e. a tree where all non-leaf vertices have degree 5) as an EMST. As a consequence, any tree of maximum degree 5 is EMST-drawable: To draw any such tree as an EMST, one can complete it to a 5-regular tree by adding any missing edges, and remove these edges from the drawing produced by \cref{alg:monma-suri}.

The algorithm performs a breadth-first search to embed the vertices by levels, where level $i$ contains all vertices at distance $i$ from an arbitrarily selected root vertex. All edges between levels $i$ and $i+1$ have the same length, denoted by $l_i$, and the angle between two consecutive neighbors at level $i+1$ is $\frac{\pi}{3}+\theta_i$, where the length $l_i$ and the angle $\theta_i$ are determined recursively as described in \cref{alg:monma-suri}.

\begin{algorithm}[H]

\DontPrintSemicolon
\SetAlgoLined
\KwIn{a $5$-regular tree $T$ rooted at node $r$}

place the root $r$ at the origin\;
place the five neighbors of $r$ evenly spaced on a unit circle centered at $r$\;
$\theta_0 \leftarrow \frac{\pi}{15}$\;
$l_0 \leftarrow 1$\;

\ForEach{vertex $v$ on level $i \ge 1$ with parent $v_0$ and children $v_1,\dots,v_4$}{
    $\theta_i \leftarrow \frac{\theta_{i-1}}{15} = \frac{\pi}{15^{i+1}}$\;
    $l_i \leftarrow \frac{l_{i-1}\theta_{i-1}}{3} = \frac{\pi^i}{3^i\cdot15^{i(i+1)/2}}$\;
    
    place $v_1,\dots,v_4$ on a circle centered at $v$ of radius $l_i$ such that:\;
    
    $\angle v_1 v v_0 = \angle v_0 v v_4 = \frac{\pi}{2} - \frac{\theta_{i-1}}{10}$\;

    $\angle v_j v v_{j+1} = \frac{\pi}{3} + \theta_i$ for $j\in\{1,2,3\}$\;
}
\KwOut{a planar drawing $\Gamma(T)$ of $T$}
\caption{Monma-Suri Algorithm}
\label{alg:monma-suri}
\end{algorithm}
\noindent See \cref{fig:monma-suri} for a visual representation of the iterative step of \cref{alg:monma-suri}.

\begin{figure}
    \centering
    \includegraphics[scale=.75]{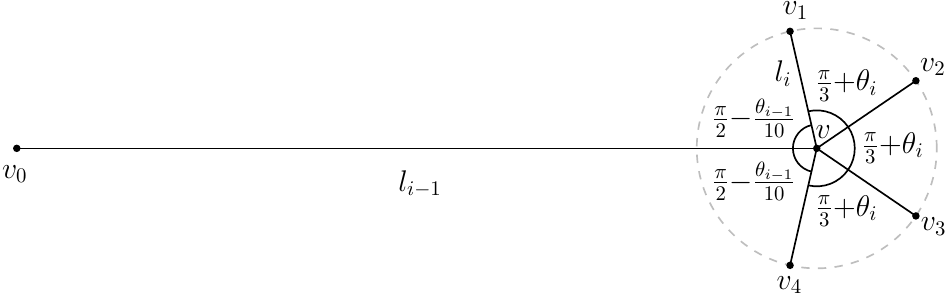}
    \caption{Drawing children $v_1,\dots,v_4$ of $v$; not drawn to scale.}
    \label{fig:monma-suri}
\end{figure}

\begin{theorem}[\cite{monma1992transitions}]\label{thm:MonmaSuri}
Let $T$ be a tree with maximum vertex degree five and $n$ vertices.  \cref{alg:monma-suri} computes an EMST-drawing of $T$ in $O(n)$ time in the real RAM model.
\end{theorem}

Observe that the drawing produced by \cref{alg:monma-suri} is actually a strong EMST-drawing.

\begin{corollary}
    \label{cor:MonmaSuriStrong}
Let $T$ be a tree with maximum vertex degree five and $n$ vertices. \cref{alg:monma-suri} computes a strong EMST-drawing of $T$ in $O(n)$ time in the real RAM model.
\end{corollary}

Throughout this paper, we often consider caterpillars. A \emph{caterpillar} is a tree such that, after removing all the leaves, the remaining graph is a path. We call this path the \emph{spine}. Note that, in particular, the first and last vertex of the spine (or any other spine vertex) are not vertices of degree one.
Let $T$ be a caterpillar whose spine consists of $t$ vertices $v_1,v_2,\dots,v_t$ appearing along the spine in this order. 

We write $T=(d_1,\dots,d_t)$ to denote the isomorphism type of $T$,
where $d_i=\deg_T(v_i)$ for each $i\in\{1,\dots,t\}$. Observe that each $d_i$ in any isomorphism type satisfies $d_i\ge 2$. 

We frequently use a subcaterpillar induced by the vertices of a subpath of the spine and all their neighbors, formally defined as follows:

\begin{definition}\label{def:subcaterpillar}
Let $T=(d_1,\dots,d_t)$ be a caterpillar with spine $v_1,\dots,v_t$. For $1\leq i\leq j\leq t$, we define $T_{ij}$ as the caterpillar $T_{ij}=(d_i,d_{i+1},\dots, d_j)$ with spine $v_i,v_{i+1},\dots,v_j$. 
We call $T_{ij}$ a \emph{spine-induced subcaterpillar} of $T$, abbreviated as \emph{spinar} from 
{\color{NavyBlue}sp}ine-{\color{NavyBlue}in}duced subcaterpill{\color{NavyBlue}ar}.
\end{definition}

In particular, if $i>1$ or $j<t$, then $v_{i-1}$ or $v_{j+1}$, respectively, is a leaf of $T_{ij}$. 

Eades and Whitesides~\cite{EadesW1996tcs,EadesW1996algo} showed that testing whether a tree is EMST-drawable is \NP-hard. When constructing the reduction gadgets, they made several observations about EMST-drawings of trees, some of which we use to prove our results.

\begin{lemma}[\cite{EadesW1996algo}]\label{lem:6uniq}
    Let $T$ be a degree-$6$ star, that is, $T=(6)$. Then, in any EMST-drawing of $T$, all edges have the same length, and consecutive edges in the cyclic ordering of edges around the center form angles of size $\frac{\pi}{3}$.
\end{lemma}

\begin{lemma}[\cite{EadesW1996algo}]\label{lem:64grid}
    Let $T$ be a tree formed by attaching three new leaves to a leaf $v$ of a star $K_{1,6}$, that is, $T=(6,4)$. Then, in any EMST-drawing of $T$, $v$ is located at the center of a regular hexagon whose vertices are occupied by the center of the star, two leaves of the star, and the three new leaves attached to $v$.
    \end{lemma}

In particular, \cref{lem:64grid} tells us that, in any EMST drawing of a caterpillar $T=(6,4)$, the vertices are placed on a triangular grid and edges are drawn as edges of the triangular grid. We will call such a drawing an \emph{embedding into a triangular grid}. See \cref{fig:644caterpillar} for an illustration. The following is an easy corollary of \cref{lem:64grid}.

\begin{corollary}\label{cor:644grid}
        Let $T=(6,4,\dots,4)$ be a caterpillar.\footnote{Whenever $4,\dots,4$ appears in a caterpillar type in this paper, it allows also one or no occurrences of 4.}
        Then any EMST-drawing of $T$ is an embedding into a triangular grid. Moreover, if $v$ is a spine vertex of $T$, then all grid points adjacent to $v$ are occupied by other vertices of $T$.
        
\end{corollary}

\begin{figure}[h]
    \centering
    \includegraphics[page=3]{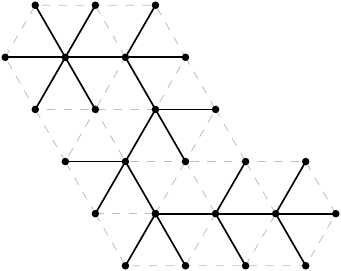}
    \caption{An embedding of $(6,4,4,4,4,4,4)$ into the triangular grid.}
 \label{fig:644caterpillar}
\end{figure}

\section{EMST-Drawable Caterpillars}\label{se:MST-caterpillars}

In this section we fully characterize caterpillars that are EMST-drawable in terms of their degree sequences. In particular, we describe a linear-time algorithm that decides whether a tree is EMST-drawable and, in the affirmative, outputs a strong EMST-drawing. We remark here that our decision algorithm is purely combinatorial, and runs in linear time even in the word RAM model of computation.

\begin{theorem}[Characterization of caterpillars that are EMST-drawable]\label{thm:char-6-EMST}
Let $T=(d_1,\dots, d_t)$ be a caterpillar with $n$ vertices. 
Then $T$ is EMST-drawable if and only if all of the following conditions are fulfilled:
\begin{itemize}
    \item For all $i\in \{1,2,\dots,t\}$, we have $ d_i\leq 6$.
    \item For all $i<j$ such that $d_i=d_j=6$, at least one of the following conditions holds:
    \begin{itemize}
        \item there exist $k,k'$ such that $i<k<k'<j$ 
        and $d_k=d_{k'}=3$, or
        \item there exists $k$ such that $i<k<j$ 
        and $d_k=2$.
    \end{itemize}
    \item For all $i<j$ such that $\{d_i,d_j\}=\{5,6\}$, there exists $k$ such that $i<k<j$ and $d_k\leq 3$.
\end{itemize}
Moreover, there is an algorithm that decides whether $T$ is EMST-drawable in $O(n)$ time  in the word RAM model. If $T$ is EMST-drawable, it outputs a strong EMST-drawing of $T$ in $O(n)$ time in the real RAM model.
\end{theorem}

\begin{proof}
We split the proof into several claims. We first prove that, if the given caterpillar $T$ does not fulfill the conditions of the theorem, then $T$ cannot be drawn as an EMST. Firstly, note that by~\cref{obs:MSTmax6}, if $T$ contains a vertex of degree at least 7, then it is not EMST-drawable.

\begin{restatable}[\restateref{claim:56cater}]{ourclaim}{claimFiveSixCater}
\label{claim:56cater}
Assume that $T$ has a spine of length $t \ge 2$. If $d_1=6$, $d_t\in \{5,6\}$ and for each $i \in \{2,\dots t-1\}$ we have $d_i = 4$, then $T$ is not EMST-drawable. 
\end{restatable}

\begin{restatable}[\restateref{cl:636cater}]{ourclaim}{claimSixThreeSixCater}
\label{cl:636cater}
    Assume that $T$ has a spine of length $t\ge3$. If $d_1=d_t=6$, $d_i=3$ for some $i\in \{2,\dots,t-1\}$ and $d_j = 4$ for all $j\in\{2,\dots,t-1\}$, $j\ne i$, then $T$ is not EMST-drawable. 
\end{restatable}

 Both \cref{claim:56cater} and \cref{cl:636cater} can be proven in the following way: Suppose that such a caterpillar $T$ has an EMST-drawing~$\Gamma$. Either remove or add leaves to $T$ in order to obtain a caterpillar containing $(6,4,\dots,4)$. The drawing $\Gamma$ will necessarily contain a drawing of $(6,4,\dots,4)$ that is not an embedding into the triangular grid, and thus contradicts \cref{cor:644grid}. See \cref{app:MST-caterpillars} for detailed proofs.

We next show that, if $T$ fulfills the assumptions of the theorem, there exists an EMST-drawing of $T$. 

First, we assume that $T$ contains at least one vertex of degree $5$. We write $T$ as a union of subcaterpillars  $T=  S_0  \cup F_1 \cup S_1 \cup F_2  \cup \cdots  \cup  S_{l-1} \cup F_l \cup S_l$ as follows. We define $F_1,F_2,\dots, F_l$ to be 
spinars of $T$ induced by segments of the spine consisting only of vertices of degree $5$. For $i\in \{1,\dots, l-1\}$, if $F_i = T_{ab}$, $F_{i+1} = T_{cd}$, we define $S_i = T_{(b+1)(c-1)}\setminus \{v_b,v_c\}$. 
If $F_1= T_{ab}$ and $a>1$, we define $S_0=T_{1(a-1)}\setminus \{v_a\}$; otherwise, we define $S_0$ to be empty. Similarly, if $F_l= T_{ab}$ and $b<k$, we define $S_l=T_{(b+1)k}\setminus\{v_b\}$; otherwise, we define $S_l$ to be empty.  See \cref{fig:decomposition} for an example of this decomposition. 

\begin{figure}[h]
    \centering
    \includegraphics[width = \linewidth, page=5]{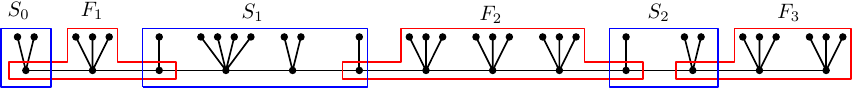}
    \caption{Decomposition of $T$ into $S_0\cup F_1 \cup S_1 \cup F_2 \cup S_2 \cup F_3 \cup S_3$, where $S_3 = \emptyset$.}
    \label{fig:decomposition}
\end{figure}

We first show how to draw $S_0,S_1,\dots, S_l$.

\begin{restatable}[\restateref{cl:S_idrawing}]{ourclaim}{claimCidrawing}
\label{cl:S_idrawing}
For every $i\in \{0,\dots,l\}$, $S_i$ admits an EMST-drawing $\Gamma$ satisfying the following conditions: 

\begin{itemize}
    \item $\Gamma$ is an embedding into a triangular grid;
    \item  if $1\le i\le l-1$, the convex hull of the vertices is contained in a parallelogram with angles $\frac{\pi}3$ and $\frac{2\pi}3$ and side lengths $2$ and $s$ (in the grid units), where $s$ is the number of vertices on the spine of $S_i$;
    \item all spine vertices lie on a common line.
\end{itemize}  
\end{restatable}

\begin{claimproof}[Proof (sketch)]

Since $T$ satisfies the conditions of \cref{thm:char-6-EMST}, so does each $S_i$. Thus, between any two consecutive vertices of degree $6$ we have either two vertices of degree $3$ or one vertex of degree $2$, and those allow us to ``switch'' between the unique representations of $(6,4,\dots,4)$ segments of the spine; see the second row in \cref{fig:T_idrawing}. Similarly, if $1\le i\le l-1$, we know that the first and last vertices in the spine of $S_i$ have degree $2$ or $3$, and hence we can ensure the second condition of the claim; see the first row of \cref{fig:T_idrawing}. 
The last row of \cref{fig:T_idrawing} shows an example of the final drawing of $S_i$. 
\end{claimproof}

    \begin{figure}[h]
        \centering
        \includegraphics[width=\linewidth, page = 13]{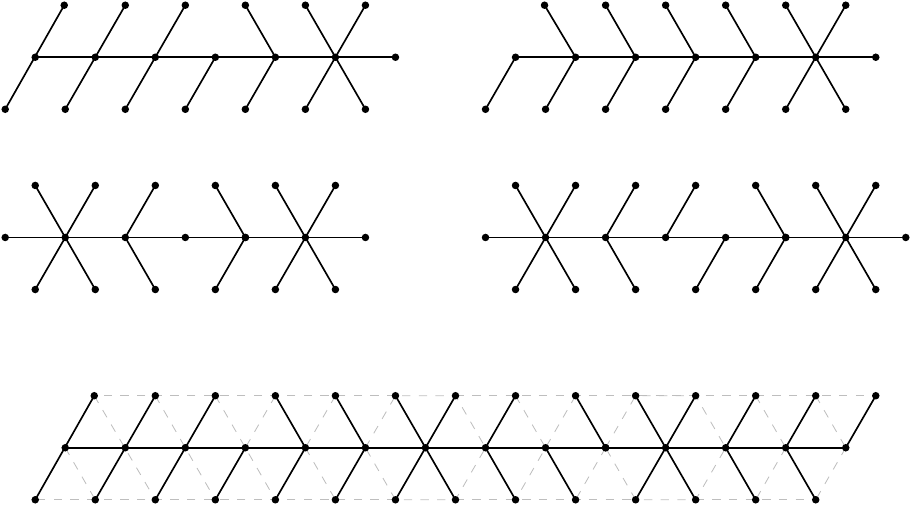}
        \caption{Drawing of $S_i$ satisfying the conditions of \cref{cl:S_idrawing}. In the last row, dashed edges are the edges of the underlying triangular grid on which we place the vertices.}
        \label{fig:T_idrawing}
       
    \end{figure}

Before continuing, we note that, if $T$ contains no vertices of degree $5$, then it has the same structure as the $S_i$'s and can be drawn in almost the same way as in the proof of \cref{cl:S_idrawing}.

Next, we describe how to draw $F_1,F_2,\dots,F_l$.

\begin{restatable}[\restateref{cl:555caterpillars}]{ourclaim}{claimFFFiveCater}
\label{cl:555caterpillars}
    Let $T = (d_1,\dots, d_t)$ be a caterpillar such that $d_1=d_2= \cdots = d_t = 5$. Then  $T$ has a strong EMST-drawing in which the spine forms an $x$-monotone curve. 
\end{restatable}

\begin{claimproof}[Proof (sketch)]
    We root $T$ at $v_1$, apply \cref{alg:monma-suri}, and, in each step, we keep track of the position of the spine vertex to ensure $x$-monotonicity.
\end{claimproof}

Now, we combine \cref{cl:S_idrawing,cl:555caterpillars} to construct a strong EMST-drawing of $T$.

We first draw $F_1= (d_1,\dots,d_s)$ according to \cref{cl:555caterpillars}. Then, if $S_0$ is nonempty, we draw it according to \cref{cl:S_idrawing}, with edges of length $l_1$ from \cref{alg:monma-suri}. Note that the last vertex of the spine of $S_0$ is contained in the drawing of $F_1$ as a leaf adjacent to $v_1$. In particular, we can choose the last vertex of the spine of $S_0$ to be the leftmost leaf adjacent to $v_1$, say $v_l$. We first join the drawing of $S_0$ with the drawing of $F_1$ so that the spine of $S_0$ has the same slope as the edge $v_1v_l$. Lastly, we rotate the drawing of $S_0$ counterclockwise around $v_l$ by an angle of $\frac\pi6$. This way, the vertices of $S_0$ lie outside the union of circles of radius $1$ centered at $v_1$ and its neighbors (different from $v_l$), see \cref{fig:c0attach}. Now we argue that there are no bad pairs between vertices of $S_0$ and $F_1$. Note that for a pair of vertices $u,v$ where $u\in S_0$, $v\in F_1$, the path between $u$ and $v$ contains an edge of $F_1$ of length $1$. On the other hand, the edges of $S_0$ are drawn with length $l_1 < 1$ and by our previous argument $d(u,v) > 1$. Hence, $u$ and $v$ do not form a bad pair. 
 
\begin{figure}
    \centering
    \includegraphics[page=1]{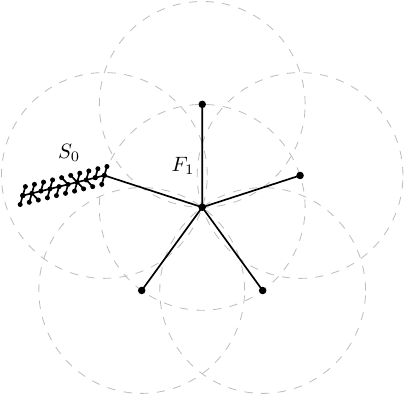}
    \caption{Attaching the drawing of $S_0$ to the drawing of $F_1$ without creating any bad pairs. The circles of radius $1$ centered at vertices of $F_1$ are drawn dashed. }
    \label{fig:c0attach}
\end{figure}

Next, assume that $i \ge 1$ and that we have a strong EMST-drawing $\Gamma_i$ of~$T_i = S_0 \cup F_1 \cup \cdots \cup S_{i-1} \cup F_i$ where the spine of $T_i$ is drawn as an $x$-monotone curve and the last vertex on the spine of $F_i$ is the rightmost vertex of the spine in $\Gamma_i$. We describe how to draw $S_i$ and $F_{i+1}$ so that the resulting drawing $\Gamma_{i+1}$ of~$T_{i+1}= S_0\cup F_1 \cup \cdots \cup  F_i \cup S_{i} \cup F_{i+1}$ is still a strong EMST-drawing with the same properties. Similar to before, the first vertex of the spine of  $S_i$ is a leaf of $F_i$; in particular, we can assume that it is the rightmost leaf $v_r$ attached to the last vertex $v$ of the spine of the current drawing $\Gamma_i$. Let $j\ge 0$ be such that the edge $vv_r$ is drawn with length $l_j$ from \cref{alg:monma-suri}. Then, we draw $S_i$ with all edges of length $l_{j+1}$ according to \cref{cl:S_idrawing}. We then attach the drawing of $S_i$ to the current drawing $\Gamma_i$ so that the two vertices representing $v_r$ in the drawing are identified, and so that the spine of $S_i$ is drawn on the line through $v$ and $v_r$. Let $\varphi$ be the slope of $vv_r$. To ensure the $x$-monotonicity of the spine, if $\varphi$ is positive (respectively, negative), we rotate the drawing of $S_i$ around $v_r$ by an angle of $\frac\pi6$ clockwise (respectively, counterclockwise). Additionally, if $\varphi$ is positive, we reflect the vertices of $S_i$ around the spine of $S_i$. Call the resulting drawing $\Gamma_i'$. Again, this process guarantees that all vertices of $S_i$ lie outside the union of circles with a diameter $l_j$ centered at $v$ and its neighbors (different from $v_r$) and, as we argued in the case of $S_0$ and $F_1$, this is sufficient to show that the vertices of $S_i$ do not form any bad pairs with the vertices of $T_i$. See \cref{fig:nobadpairs}.

\begin{figure}
    \centering
    \begin{minipage}[c]{0.35\textwidth}
        \centering
        \includegraphics[page=6]{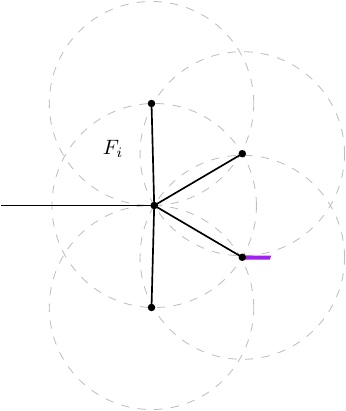}
    \end{minipage}
    \hspace{0.15\textwidth}
    \begin{minipage}[c]{0.35\textwidth}
        \centering
        \includegraphics[page=7]{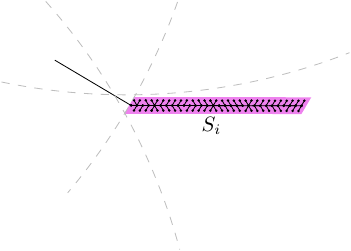}
    \end{minipage}  
    \caption{Drawing $S_{i}$ inside the purple parallelogram guarantees that no vertices of $S_{i}$ form bad pairs with vertices of $F_i$ since they lie outside all of the circles of radius $l_j$ centered at the vertices of $F_i$.  
    The right part  shows the zoomed in situation and an example drawing of $S_{i}$.}
    \label{fig:nobadpairs}
\end{figure}

Lastly, we explain how to draw $F_{i+1}$. Note that the last vertex $v_r'$ on the spine of $S_i$ (which is also the rightmost spine vertex of $\Gamma_i'$) is contained in $F_{i+1}$, as a leaf attached to the first spine vertex $v'$ of $F_{i+1}$. We draw $F_{i+1}$ according to \cref{cl:555caterpillars}. Let $v_l'$ be the leftmost leaf attached to $v'$ in the drawing of $F_{i+1}$. We attach the drawing of $F_{i+1}$ to $\Gamma_i'$ by identifying $v_l'$ and $v_r'$.  We then rotate the drawing of $F_{i+1}$ around $v_l$ so that the edge $v_l'v'$ has slope $\varphi$ (the same slope as the edge $vv_r$ in $\Gamma_i$), the drawing obtained by this is our desired drawing $\Gamma_{i+1}$. As we argued in the previous cases, this rotation guarantees that there are no bad pairs between the vertices of $F_{i+1}$ and $S_{i}$, and therefore that the  vertices of $T_i \cup S_i$ form no bad pairs. Finally, in \cref{fig:545}, it can be seen that even in the extremal case, when the spine of $S_i$ is only a single vertex, in the resulting drawing $\Gamma_{i+1}$ there are no bad pairs between the vertices of $F_{i+1}$ and the vertices of $F_{i}$ (and hence with the vertices of $T_i$). If the spine of $S_i$ consists of more than a single vertex, in $\Gamma_{i+1}$, the distance between vertices in $F_{i+1}$ and vertices in $T_i$ is increasing, so we do not introduce any bad pairs, and the resulting drawing $\Gamma_{i+1}$ is indeed an EMST-drawing of $T_i$. The spine of $\Gamma_{i+1}$ being $x$-monotone and $\Gamma_{i+1}$ being a strong EMST-drawing is guaranteed by \cref{cl:S_idrawing,cl:555caterpillars}. \qedhere

\begin{figure}[h]
    \centering
    \begin{minipage}[c]{0.58\textwidth}
        \centering
        \includegraphics[page=10, scale = 0.88]{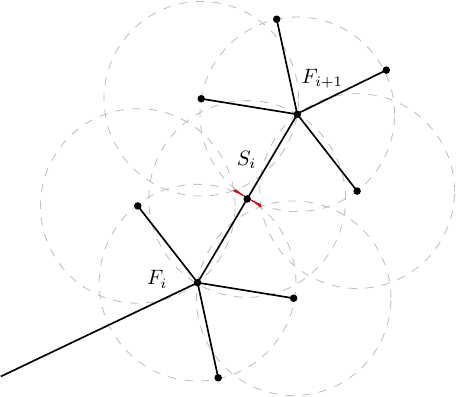}
    \end{minipage}
    \hfill
    \begin{minipage}[c]{0.38\textwidth}
        \centering
        \includegraphics[page=11,scale = 0.88]{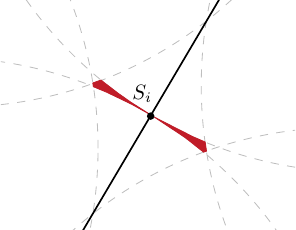}
    \end{minipage}  
    \caption{Drawing $F_{i+1}$ when $S_i$ consists of one spine vertex 
    when \cref{alg:monma-suri} was used for more than one step.  The regions where an additional leaf can be drawn, which exist because the corresponding angles are larger than $\frac{\pi}{3}$, are highlighted in red and zoomed in on the right.}
    \label{fig:545}
\end{figure}

\end{proof}

\section{The Bad Pair Number of Caterpillars of Max Degree Six}\label{se:caterpillars-bp-minimization}

Recall that \cref{thm:MonmaSuri} implies $\bpn(T)=0$ for caterpillars of maximum degree at most 5.

Let $T$ be a caterpillar of maximum degree 6 and such that $T$ does not satisfy the conditions of \cref{thm:char-6-EMST}. In this section, we present a linear-time algorithm to compute a planar drawing of $T$ whose number of bad pairs is minimized.

\begin{definition}\label{def:reducedmodel}
Let $T=(d_1,\dots,d_t)$ be a caterpillar such that $d_i\le 6$ for every $i\in\{1,\dots,t\}$.
The \emph{reduced model} $\widetilde{T}$ of $T$ is the caterpillar with vertices $u_1,\dots,u_s$ such that there exists a sequence $1=i_1<i_2<\dots<i_s=t$ with the following properties:
\begin{itemize}
\item $\deg_{\widetilde{T}}(u_j)=d_{i_j}$ for every $j\in\{1,\dots,s\}$,
\item $d_{i_j}\in\{5,6\}$ for every $j\in\{2,\dots,s-1\}$,
\item for every $j\in\{1,\dots,s-1\}$ and every $h$ such that $i_j<h<i_{j+1}$ it holds that $d_h\le 4$.
\end{itemize}
\end{definition}

Consider the reduced model $\widetilde{T}$ of $T$ with spine $v_1,\dots,v_t$ and any spinar $S$ of $\widetilde{T}$.
Observe that when $S$ has spine $u_j,\dots,u_{j'}$, then the spinar of $T$
with spine $v_{i_j},\dots,v_{i_{j'}}$ has reduced model $S$.
In the next proposition we denote this spinar by $S'$. In particular, we have $\widetilde{S'}=S$.

\begin{restatable}[\restateref{prop:disjointsubcaterpillars}]{proposition}{propdisjointsubcaterpillars}

\label{prop:disjointsubcaterpillars}
Let $\widetilde{T}$ be the reduced model of a caterpillar $T=(d_1,\dots,d_t)$ with spine $v_1,\dots,v_t$ and let $S_1,S_2$ be spinars of $\widetilde{T}$ with the corresponding spinars $S_1',S_2'$ in $T$. Let $\Gamma$ be any drawing of $T$.  
\begin{enumerate}
    \item \label{i1}
    If $S_1$ and $S_2$ do not share any spine vertex of $\widetilde{T}$, then no bad pair in $\Gamma$ belongs to both $S_1'$ and $S_2'$.
    \item \label{i2}
    If $S_1$ and $S_2$ share exactly one spine vertex of $\widetilde{T}$, say $u_h$, then every bad pair in $\Gamma$ that belongs to both $S_1'$ and $S_2'$ involves a leaf adjacent to $v_{i_h}$, except possibly the bad pair $v_{i_h-1}$ and $v_{i_h+1}$.
\end{enumerate}

\end{restatable}

See \cref{fig:subcaterpillars} for illustrations of both cases.

\begin{figure}[h]
    \centering
    \includegraphics[width=\linewidth]{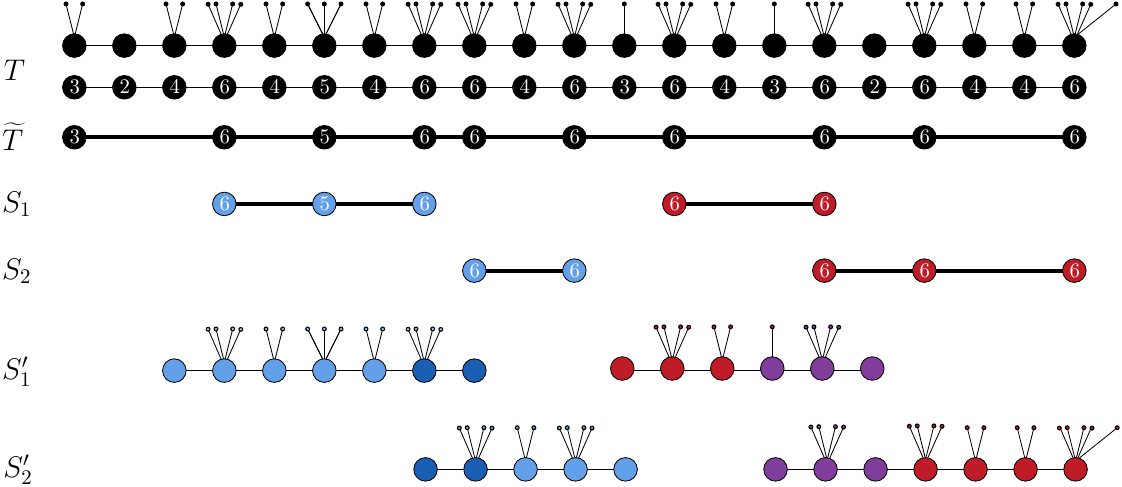}
    \caption{An illustration to the proof of \cref{prop:disjointsubcaterpillars}. In blue, we see spinars $S_1$ and $S_2$ of $\widetilde{T}$ that are disjoint as in part \ref{i1}. In red, they share one spine vertex as in part \ref{i2}. The shared vertices of $S_1'$ and $S_2'$ are marked in dark blue and purple, respectively.}
    \label{fig:subcaterpillars}
\end{figure}

Let us now describe the algorithm.  We  first construct the reduced model of the input caterpillar. Then we color and possibly orient its edges according to EMST-drawability of the spinars they induce in the input caterpillar. An edge is colored green when the spinar is EMST-drawable, and it is colored blue, purple, or red depending on which case of \cref{thm:char-6-EMST} is violated. Next, we recolor some blue and some purple edges to light blue (light purple, respectively) depending on the relative position of the edges. Finally, we color some of the vertices red. For technical details, see \cref{alg:red_vertices}. Our goal is to prove that these red vertices correspond exactly to the bad pairs in an optimal drawing. The algorithm is illustrated in \cref{fig:reduced-model}.

\begin{algorithm}[H]
\DontPrintSemicolon
\SetAlgoLined
\setcounter{AlgoLine}{0}
\KwIn{a max-degree-6 caterpillar $T=(d_1,\dots,d_t)$}
construct the reduced model $\widetilde{T}$ of $T$ \tcc*{Phase 1 - preparation  }
let its vertices be $u_1 = v_{i_1}, u_2 = v_{i_2}, \dots, u_s = v_{i_s}$\;
\For{$j \leftarrow 1$ \KwTo $s-1$ \tcc*{Phase 2 - coloring {} {} {} }}{
    \uIf{$d_{i_j} = d_{i_{j+1}} = 6$ and $d_h = 4$ for all $h$ with $i_j < h < i_{j+1}$}{
        color edge $u_j u_{j+1}$ red\;
    }
    \uElseIf{$d_{i_j} = d_{i_{j+1}} = 6$ and $d_h = 4$ for all $h$ except exactly one $h'$ with $i_j < h' < i_{j+1}$ and $d_{h'} = 3$}{
        color edge $u_j u_{j+1}$ dark blue\;
    }
    \uElseIf{$\{d_{i_j}, d_{i_{j+1}}\} = \{5,6\}$ and $d_h = 4$ for all $h$ with $i_j < h < i_{j+1}$}{
        color edge $u_j u_{j+1}$ dark purple\;
        orient it from the vertex of degree 5 to the vertex of degree 6\;
    }
    \Else{
        color edge $u_j u_{j+1}$ green\;
    }
}
\For{$j \leftarrow 1$ \KwTo $s$  \tcc*{Phase 3 - recoloring {} }}{
    \If{$u_j$ is incident with a red edge in $\widetilde{T}$}{
        color $u_j$ red\;
        recolor every dark purple edge incident with $u_j$ light purple\;
        recolor every dark blue edge incident with $u_j$ light blue\;
    }
}

\For{$j \leftarrow 1$ \KwTo $s$  \tcc*{Phase 4 - coloring vertices {} }}{
    \If{$u_j$ is the head of a dark purple edge}{
        color $u_j$ red\;
        \lIf{$u_j$ is incident with a dark blue edge}{
        recolor the edge light blue
        }
    }
}

\ForEach{maximal dark blue path $u_j, \dots, u_{j+h}$}{
    \lFor{$k \leftarrow 0$ \KwTo $\left\lfloor \frac{h-1}{2} \right\rfloor$}
    {color $u_{j+1+2k}$ red}
}

\KwOut{the number $\mathrm{rv}(T)$ of red vertices of $\widetilde{T}$}
\caption{Coloring of Reduced Caterpillar}\label{alg:red_vertices}
\end{algorithm}

\begin{theorem}\label{thm:bpn-caterpillarmax6}
For every caterpillar $T$ of maximum degree at most 6, $\bpn(T)=\rv(T)$.
\end{theorem}

\begin{proof}
Note first that if $\rv(T)=0$, that means that the algorithm did not color any vertex red, so there can be no red edges. Therefore, all dark purple edges would stay dark purple and each would produce one red vertex, and thus there can be no dark purple edges in $\widetilde{T}$. But then all dark blue edges would stay dark blue, and every dark blue path would produce at least one red vertex, and hence, all edges of $\widetilde{T}$ are green. Then \cref{thm:char-6-EMST} implies that $T$ is EMST-drawable, i.e., $\bpn(T)=0$. 

\begin{figure}
    \centering
    \includegraphics[width=\linewidth]{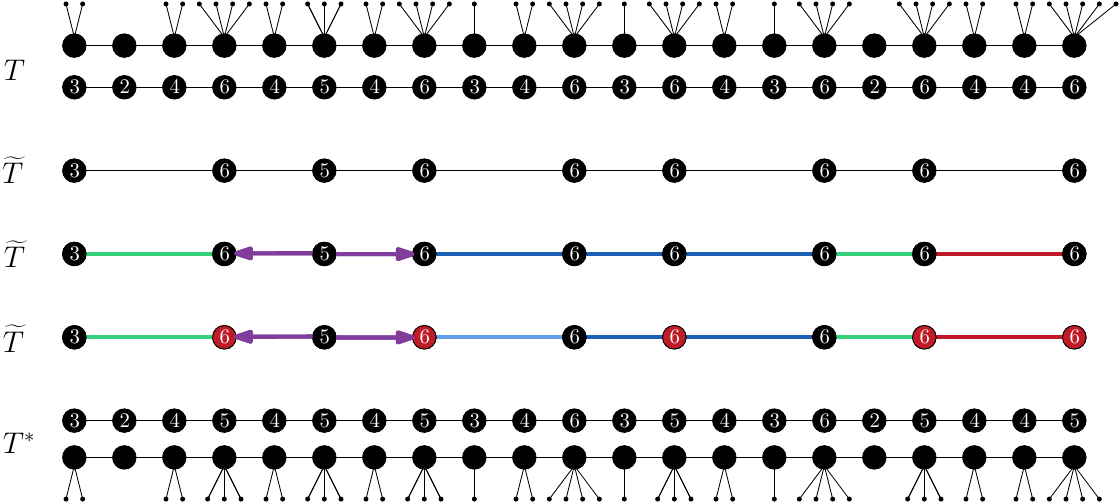}
    \caption{An example of the construction and coloring of the reduced model $\widetilde{T}$ of a caterpillar $T$, and the subsequent caterpillar $T^*$ constructed in the proof of \cref{thm:bpn-caterpillarmax6}. The first two rows show a caterpillar $T$---as a graph and as encoded in the sequence of degrees of its spine vertices. The next row shows the reduced model $\widetilde{T}$, further down is $\widetilde{T}$ after Phase~2 of  \cref{alg:red_vertices}, and then again $\widetilde{T}$ after Phase~3. The last two rows show the corresponding $T^*$.
    }
    \label{fig:reduced-model}
\end{figure}

Now we shall prove that $\bpn(T)\le\rv(T)$. Let $\widetilde{T}$ be colored by \cref{alg:red_vertices}. Consider the caterpillar $T^*=(d^*_1,\dots,d^*_t)$ defined as follows:
\begin{itemize}
    \item 
$d^*_{i_j}=d_{i_j}-1$  if $u_j$  is red in $\widetilde{T}, j\in\{1,\dots,s\}$, and
    \item 
$d^*_{h}=d_{h}$ for all other $h\in \{1,\dots,t\}$. 
\end{itemize}
Apply \cref{alg:red_vertices} to $T^*$. All edges of $\widetilde{T^*}$ will become green, and thus $T^*$ has an EMST-drawing. 

\cref{thm:char-6-EMST} implies that $T^*$ has a drawing $\Gamma$ with no bad pairs and such that every vertex of degree 5 has a leaf with a safe angle. If the degree of a vertex $v_h$ was reduced (when going from $T$ to $T^*$), then $h=i_j$ for some $j\in\{1,\dots,s\}$ and $u_j$ was red in $\widetilde{T}$, which means $d_{i_j}=6$ and $d^*_{i_j}=5$. Take the drawing $\Gamma$ and for every such $h=i_j$, choose a leaf $\ell$ adjacent to $v_{i_j}$ with a safe angle and add a new leaf $\ell'$ adjacent to $v_{i_j}$ such that the lengths of $v_{i_j}\ell$ and $v_{i_j}\ell'$ are equal and  the angle between the leaf edge $v_{i_j}\ell$ and the new edge $v_{i_j}\ell'$ is smaller than the safe angle. Each such added leaf creates exactly one bad pair, so we have a drawing of $T$ with $\rv(T)$ bad pairs.

The inequality $\bpn(T)\ge\rv(T)$ follows from the following series of claims,
which are proven by estimating the number of bad pairs caused by breaking the different conditions of \cref{thm:char-6-EMST}.

\begin{restatable}[\restateref{cl:bpn-catmax6-1}]{ourclaim}{claimCatmaxSixA}

\label{cl:bpn-catmax6-1}
If $\widetilde{T}$ has exactly two vertices and forms a red edge, then $\bpn(T)\ge 2$.    

\end{restatable}

\begin{restatable}[\restateref{cl:bpn-catmax6-2}]{ourclaim}{claimCatmaxSixB}

\label{cl:bpn-catmax6-2}
    If $\widetilde{T}$ has exactly three vertices and forms a path consisting of two red edges, then $\bpn(T)\ge 3$.

\end{restatable}

\begin{restatable}[\restateref{cl:bpn-catmax6-3}]{ourclaim}{claimCatmaxSixC}

\label{cl:bpn-catmax6-3}
    If $\widetilde{T}$ only has red edges, then $\bpn(T)\ge s$.

\end{restatable}

\begin{restatable}[\restateref{cl:bpn-catmax6-4}]{ourclaim}{claimCatmaxSixD}

\label{cl:bpn-catmax6-4}
    If $\widetilde{T}$  consists of one dark purple edge, then $\bpn(T)\ge 1$.

\end{restatable}

\begin{restatable}[\restateref{cl:bpn-catmax6-5}]{ourclaim}{claimCatmaxSixE}

\label{cl:bpn-catmax6-5}
    If $\widetilde{T}$ consists of two dark purple edges with the common vertex of degree~5, then $\bpn(T)\ge 2$.

\end{restatable}

\begin{restatable}[\restateref{cl:bpn-catmax6-6}]{ourclaim}{claimCatmaxSixF}

\label{cl:bpn-catmax6-6}
    If $\widetilde{T}$  consists of only dark purple edges, then $\bpn(T)\ge \rv(T)$.

\end{restatable}

\begin{restatable}[\restateref{cl:bpn-catmax6-7}]{ourclaim}{claimCatmaxSixG}

\label{cl:bpn-catmax6-7}
    If $\widetilde{T}$  is a path of $h$ dark blue edges, then $\bpn(T)\ge \floor{\frac{h+1}{2}} = \rv(T)$.

\end{restatable}

\medskip
To wrap up the proof, observe that after the spine edges (and vertices) of $\widetilde{T}$ are colored by \cref{alg:red_vertices}, the inclusion-max imal red,  dark purple  and  dark blue paths are  pairwise vertex disjoint. Furthermore, \Cref{cl:bpn-catmax6-3,cl:bpn-catmax6-6,cl:bpn-catmax6-7} imply that the spinars of $T$ corresponding to these paths produce at least as many bad pairs as the numbers of red vertices in these components. \cref{prop:disjointsubcaterpillars}, part \ref{i1} implies that these bad pairs are distinct, and thus the total number of bad pairs in any drawing of $T$ is bounded from below by $\rv(T)$.  
\end{proof}

From the proof of \cref{thm:bpn-caterpillarmax6} and the fact that both the algorithms produced by \cref{thm:char-6-EMST} and \cref{alg:red_vertices} work in linear time, we obtain the following corollary.

\begin{corollary}\label{cor:max-6-algo}
There is an algorithm that takes as input an $n$-vertex caterpillar $T$ of maximum degree at most 6 and produces a planar drawing $\Gamma$ with $\bpn(\Gamma)=\bpn(T)$ in $O(n)$ time in the real RAM model.
Moreover, $\bpn(T)$ can be computed in $O(n)$ time even in the word RAM model.
\end{corollary}

\section{Beyond Degree Six}\label{se:beyond}

In this section, we consider trees $T$ with maximum degree at most $\Delta$ and the problem of computing planar drawings of $T$ with $\bpn(T)$ bad pairs. 
We compute $\bpn(T)$ in the case when $T$ is a star and, for general trees of maximum degree $\Delta$, we give a $\Delta^2n\log n$ upper bound.

\subsection{Stars}\label{se:stars}
We first consider drawings of stars and, for every positive integer $s$, we determine exactly the number $\bpn(K_{1,s})$ (see \cref{thm:star_bpn}).

We say that a drawing $\Gamma$ of a star $K_{1,s}$ is \emph{pentagonal} when the leaves of the star can be partitioned into five groups so that any two leaves in the same group form a bad pair and any two leaves in different groups form a good pair.
Observe that a drawing of $K_{1,5}$ in which leaves are drawn in the vertices of a rectangular pentagon and the center vertex in the center of the pentagon is pentagonal, and each leaf has a safe angle $\frac{\pi}{15}$.
Thus, we can obtain a pentagonal drawing of a star $K_{1,s}$ by first drawing every leaf in one of the vertices of a pentagon and subsequently rotating the edge around the center by a unique angle smaller than $\frac{\pi}{30}$.

Let $t(n,r)$ denote the maximum number of edges in a graph on $n$ vertices that does not contain $K_{r+1}$ as a subgraph.
For a complementary notion, let $t'(n,r)$ denote the minimum number of edges in a graph on $n$ vertices such that every induced subgraph on $r+1$ vertices has at least one edge.
Observe that $t'(n,r) = \binom{n}{2} - t(n,r)$.
Turán's theorem~\cite{turanthm} gives that the unique extremal graph with $t(n,r)$ edges and no subgraph isomorphic to $K_{r+1}$ is a complete $r$-partite graph with parts differing in size by at most one.
As a consequence, we get the following equations.
\begin{align*}
   t'(0,r) &= 0\text{,}
\\ t'(n,r) &= t'(n-1,r) + \floor{\frac{n-1}{r}}\text{,}    
\\ t'(n,r) &= r\cdot\binom{\floor{n/r}}{2}+\floor{\frac{n}{r}}(n \bmod r)\text{.}
\end{align*}

\begin{restatable}[\restateref{obs:star_pentagonal}]{observation}{obsStarPentagonal}

\label{obs:star_pentagonal}
    For every positive integer $s$, there is a pentagonal drawing~$\Gamma$ of $K_{1,s}$ with $\bpn(\Gamma)=t'(s,5)$.

\end{restatable}

\begin{theorem}\label{thm:star_bpn}
   For every $s \in \mathbb{N}$, we have
    \[
        \bpn(K_{1,s}) = \begin{cases} 0 & \text{ for } s\le 6\text{,}\\
            t'(s,5)& \text{ for } s> 6 \text{.}
        \end{cases}
           \]
           where $t'(s,5)=5\cdot \binom{\lfloor s/5 \rfloor}{2}+\floor{\frac{s}{5}}(s \bmod 5) \approx \frac{1}{10} \cdot s^2 +\Oh{s}$.
\end{theorem}

\begin{proof}
    The value $\bpn(K_{1,6})=0$ is certified by a drawing in which the center of the star is represented in the origin, and the leaves of the star are represented by six vertices of a regular hexagon with the center of the hexagon in the origin.

    By \cref{obs:star_pentagonal} we know that $\bpn(K_{1,s})\le t'(s,5)$.
    It remains to show that for every positive integer $s \neq 6$, we have $\bpn(K_{1,s})\ge t'(s,5)$.
    For any drawing~$\Gamma$ of $K_{1,s}$, let $G(\Gamma)$ be a graph defined on the vertex set being the set of leaf vertices in~$K_{1,s}$.
    In the graph~$G(\Gamma)$ there is an edge joining vertices $u_1$ and $u_2$ if and only if $u_1$, $u_2$ is a good pair in the drawing~$\Gamma$.
    If $G(\Gamma)$ is $K_6$-free, then by the Tur\'an Theorem, there are at most $t(s,5)$ edges in $G(\Gamma)$,
    the number of good pairs in the drawing $\Gamma$ is at most $t(s,5)$, and $\bpn(\Gamma) \ge t'(s,5)$.

    Thus, it remains to show that for every $s\ge 7$ and a drawing $\Gamma$ of $K_{1,s}$ with the clique number of $G(\Gamma)$ at least $6$ we have $\bpn(\Gamma)\ge t'(s,5)$.
    When $G(\Gamma)$ contains a clique $K_6$, then the vertices of the clique are represented as vertices of a regular hexagon in $\Gamma$ and the central vertex of the star is represented as the center of this hexagon.
    It follows, that every other leaf vertex creates at least two bad pairs with the vertices of the clique $K_6$, and
    $\bpn(\Gamma) \ge 2(s-6)$.

    The proof goes by induction on $s$.
    For $s\ge7$ and $s\le10$ observe that $t'(s,5) = s-5 \le 2(s-6)$.
    For $s\ge11$ and $s\le12$ observe that $t'(s,5) = 5 + 2(s-10) \le 2(s-6)$.
    
    Now, let $s \ge 13$ and $U$ be the set of six vertices that span a clique $K_6$ in $G(\Gamma)$.
    By the induction hypothesis applied to $G\setminus U$ there are at least $t'(s-6,5)$ bad pairs in the drawing~$\Gamma$ among the vertices not in $U$.
    Further, every vertex not in $U$ creates at least  two bad pairs with vertices in $U$.
    We get
    \[
        \bpn(\Gamma) \ge t'(s-6,5) + 2(s-6)
        \text{.}
    \]
    
    On the other hand, we have
    \begin{align*}
        t'(s,5) &= t'(s-1,5)+\floor{\tfrac{s-1}{5}} = t'(s-2,5)+\floor{\tfrac{s-2}{5}}+\floor{\tfrac{s-1}{5}}\\
        &= t'(s-6,5)+\floor{\tfrac{s-6}{5}}+\floor{\tfrac{s-5}{5}}+\floor{\tfrac{s-4}{5}}+\floor{\tfrac{s-3}{5}}+\floor{\tfrac{s-2}{5}}+\floor{\tfrac{s-1}{5}}
        \\    &\le t'(s-6,5)+6\cdot\tfrac{s}{5} - \tfrac{1}{5}-\tfrac{2}{5}-\tfrac{3}{5}-\tfrac{4}{5}= t'(s-6,5)+\tfrac{6}{5}s - 2
        \text{.} \end{align*}
    For $s\ge13$ we have that $\frac{6}{5}s-2 \le 2s-12$, and we get $\bpn(\Gamma) \ge t'(s,5)$.
    \end{proof}

\subsection{Trees} \label{se:trees}

In this section we consider trees of maximum degree $\Delta \ge 6$. 
Note that for any tree $T$, it is clear that $\bpn(T) \le n^2$. The main result of this section is a construction of a drawing that contains a subquadratic number of bad pairs if $\Delta=o(\sqrt{n/\log n})$.

\begin{theorem}\label{thm:max-D-ary-tree}

    Let $T$ be a tree of maximum degree $\Delta$. Then, $\bpn(T)<\Delta^2 n\log n$.
\end{theorem}

\begin{proof}
    Assume that $\Delta\geq 6$, as otherwise $\bpn(T)=0$.
    We construct a drawing $\Gamma(T)$ of $T$ in which $\bpn(\Gamma(T))$ is bounded by $cn\log n$, where $c=\frac{\Delta-1}{\log\left(\frac{\Delta}{\Delta-1}\right)}$ and $\log$ is the base two logarithm. Note that this suffices, as we have $\frac{1}{\Delta-1}<1$, so $\log\bigl(1+\frac{1}{\Delta-1}\bigr)>\frac{1}{\Delta-1}$, and thus $c<\Delta^2 $.
    
    We will show the claim by induction on the height of the tree.
       The statement holds trivially for trees of height one (i.e. stars), as the number of bad pairs is at most $\Delta^2 < \Delta n < cn \log n$.
    
    Let $T$ be a tree of maximum degree $\Delta$, rooted at an arbitrarily chosen non-leaf vertex $v_1$. Let $v_1,\dots, v_n$ be the list of vertices of $T$ sorted according to DFS preordering (i.e. according to the first time they were visited by DFS), such that in each iteration, we visit the unvisited child with the smallest number of descendants (breaking ties arbitrarily).
    Let $v_{j_1},\dots, v_{j_k}$ be the children of $v_1$, such that $j_1<\dots<j_k$. For $i\in\{1,\dots,k\}$, we denote by $T_i$ the subtree of $T$ rooted at $v_{j_i}$, and let $n_i=|V(T_i)|$. In particular, we have $n_1\leq \dots\leq n_k$.

    Let us describe the drawing $\Gamma(T)$.
    Let $\varepsilon\in \left[0, \frac1n\right]$ be a constant depending on $n$. We place the vertex $v_i$ at the point $(3^i, \varepsilon \cdot i)$. See \cref{fig:general-trees}.
    It is easy to see that $\Gamma$ is a plane straight-line drawing of $T$. 
    We use $\Gamma_i$ to denote the subdrawing of $T_i$ in $\Gamma(T)$. 

\begin{figure}[h]
    \centering
    \includegraphics[scale=.75]{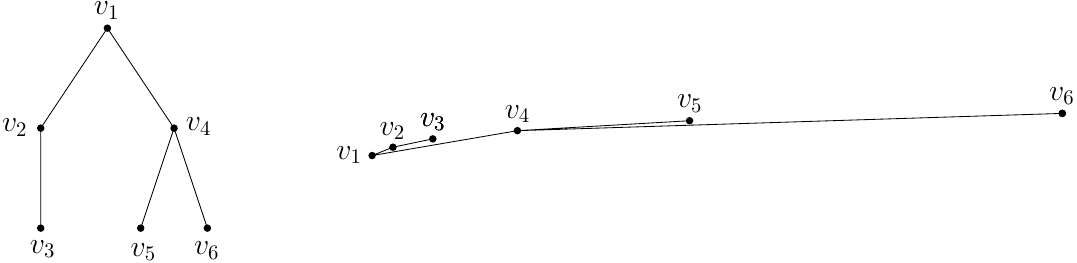}
    \caption{An example of a tree and its embedding; for readability purposes, the base of the exponent used for the $x$-coordinates is smaller than 3.}
    \label{fig:general-trees}
\end{figure}

     We now determine the bad pairs in this drawing. It is easy to see that $v_1$ does not form a bad pair with any vertex. 
    Next we consider the number of bad pairs inside $\Gamma_i$, for $i\in \{1,\dots,k\}$. 
    By induction, the drawing $\Gamma(T_i)$ has at most $c n_i\log n_i $ bad pairs.
    We can argue that $\bpn(\Gamma(T_i))=\bpn(\Gamma_i)$ (informally, the $x$ coordinates in these drawings are the same up to scaling by a factor $3^{j_i}$, and the $y$ coordinates are negligibly small).

    Next, we determine the number of bad pairs between different subtrees $T_i,T_j$ where $i \neq j$.
    Let $v_a$ and $v_b$ ($a<b$) be vertices that belong to subtrees $T_{a'}$ and $T_{b'}$ for some $a'\neq b'$, $a',b'\in \{1,\dots, k\}$. Recall that given two points, $p_1=(x_1,y_1)$ and $p_2=(x_2,y_2)$, we have $d(p_1,p_2)\leq |x_1-x_2|+|y_1-y_2|$, $d(p_1,p_2)\geq |x_1-x_2|$ and $d(p_1,p_2)\geq |y_1-y_2|$.
    Thus we have $d(v_a, v_b)\geq  3^b-3^a>3^a$. 
    The unique path in $T$ between $v_a$ and $v_b$ goes through $v_1$. On the path from $v_a$ to $v_1$, all vertices have smaller index than $a$, and thus all edges are shorter than $3^a$; in particular, they are shorter than $d(v_a,v_b)$. Let us now look at the path from $v_1$ to $v_b$. We distinguish two cases:
    \begin{itemize}
        \item \textbf{Case 1}: If $b'=b$, the edge $v_1v_b$ has length $3^b>d(v_a,v_b)$, so $v_a$ and $v_b$ form a bad pair. See, for example, $v_3$ and $v_4$ in \cref{fig:general-trees}.
        \item \textbf{Case 2}: If $b'\neq b$, by the definition of DFS preorder, we have $b'<b$. 
        The edge $v_1 v_{b'}$ has length at most $ 3^{b'}+\varepsilon b'<3^{b'}+1\leq 3^b-3^a\leq d(v_a,v_b)$.
        Let $v_cv_d$ be an edge on the path from $v_{b'}$ to $v_b$, with $c<d$. Note that $a<b'\leq c<d\leq b$, so we have $d(v_c,v_d)\leq 3^d-3^c+1\leq 3^b-3^a$. Thus, $v_a$ and $v_b$ do not form a bad pair. See, for example, $v_3$ and $v_5$ in \cref{fig:general-trees}.
    \end{itemize}
   
    In summary, $v_a$ and $v_b$ form a bad pair only if $v_b$ is a child of $v_1$ that comes after $v_a$ in the DFS ordering. 
    The total number of such pairs is at most $(k-1)n_1+(k-2)n_2+\dots+n_{k-1}$. In consequence, 
    \[
    \bpn(\Gamma(T))\leq \sum_{i=1}^k \bpn(\Gamma_i)+\sum_{i=1}^k (k-i)n_i \leq \sum_{i=1}^k cn_i\log n_i+\sum_{i=1}^k (k-i)n_i\text{.}
    \]

    It is sufficient to show that the right hand side is bounded by $cn \log n$.
    
    Let $m = n_1 + n_2 + \cdots + n_{k-1} = n-1- n_k$. Since the $n_i$ are sorted, we have $n_k \geq \frac{n-1}{k}$ and so $m \leq \frac{(n-1)(k-1)}{k} \leq \frac{n(k-1)}{k}$.
    Thus, $\sum_{i=1}^k (k-i)n_i\leq(k-1)m$, and moreover $(k-1)m\le cm\log\bigl(\frac{k}{k-1}\bigr)\leq cm\log\bigl(\frac{n}{m}\bigr)$, because $c \geq \frac{k-1}{\log\left(\frac{k}{k-1}\right)}$ for every $1<k \leq \Delta$. Therefore, we can write:

    \begin{align*}
    \sum_{i=1}^k cn_i\log n_i+\sum_{i=1}^k (k-i)n_i &= cn_k\log n_k +\sum_{i=1}^{k-1} (cn_i\log n_i) +  \sum_{i=1}^k (k-i)n_i\\
    &< cn_k\log n +\sum_{i=1}^{k-1} (cn_i\log m) + \sum_{i=1}^k (k-i)n_i\\
    &= cn_k\log n + cm\log m +\sum_{i=1}^k (k-i)n_i\\
    &\leq cn_k\log n + cm\log m + cm\log \left(\frac{n}{m}\right)\\
    &=cn_k\log n+cm \log n < cn\log n\text{.} \qedhere
    \end{align*}

\end{proof}

Using the above approach, we can get the following bound:
\begin{restatable}[\restateref{cor:boundeddepth}]{proposition}{corBoundeddepth}
\label{cor:boundeddepth}
Let $T$ be a rooted tree of height $h$ and maximum degree $\Delta$. Then, $\bpn(T)\le h\Delta n $. 
\end{restatable}

Next, we use \cref{thm:max-D-ary-tree} and, more particularly, \cref{cor:boundeddepth} to give a better bound for $\bpn(T)$ in the special case where $T$ is a $k$-caterpillar. A \emph{$k$-caterpillar} is a tree containing a central path $P$ (\emph{spine} of $T$) such that every vertex not in $P$ is within distance $k$ from some vertex on $P$. In particular, $1$-caterpillar corresponds to the usual notion of a caterpillar, and a $2$-caterpillar is frequently called a \emph{lobster}. 

\begin{proposition}\label{prop:hcaterpillars} 
Let $T$ be a $k$-caterpillar of maximum degree $\Delta$. Then, $\bpn(T) \le k\Delta n$. 
\end{proposition}

\begin{proof}
    Let $P=v_1,v_2,\dots, v_l$ be the spine of $T$. Note that, for each vertex $v_i$, $T$ contains a subtree $T_i$ of height at most $k$ rooted at $v_i$. Denote by $n_i$ the number of vertices in tree $T_i$. We construct a drawing $\Gamma$ of $T$ as follows. Firstly, we place $v_1$ at $(0,0)$. Then, for each $i\ge1$, we draw the edge $v_iv_{i+1}$ as a segment of length $3^{\Delta^{h+1}}$ and slope $-1$ if $i$ is odd, and $+1$ if $i$ is even. Then, at each vertex $v_i$, we choose a small angle $\varphi_i>0$ at $v_i$ such that $\varphi_i$ is bisected by the upwards ray originating at $v_i$ if $i$ is odd, and by the downwards ray originating at $v_i$ if $i$ is even. Now we complete our drawing $\Gamma$ of $T$ by drawing each $T_i$ as in \cref{thm:max-D-ary-tree} inside the section of the plane determined by $\varphi_i$.\footnote{Formally, we first draw $T_i$ and then apply a rotation to the drawing to place it inside the region.} See \cref{fig:lobsterdrawing}.

    \begin{figure}
        \centering
        \includegraphics{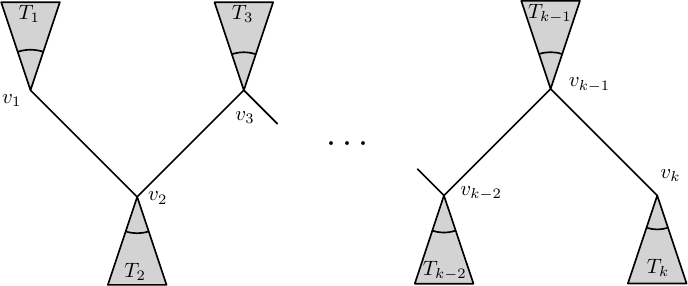}
        \caption{Drawing of a $k$-caterpillar $T$ described in the proof of \cref{prop:hcaterpillars}.} 
        \label{fig:lobsterdrawing}
    \end{figure}
    Note that each $T_i$ is a tree with at most $n_i$ vertices, and therefore it contributes at most $\bpn(T_i)\le k\Delta n_i$ bad pairs to $\bpn (\Gamma)$ by \cref{cor:boundeddepth}. Additionally, for any $i$, the longest edge in the drawing of $T_i$ is drawn with a length strictly less than $3^{\Delta^{k+1}}$. Thus, if we choose all angles $\varphi_i$ to be sufficiently small, there are no bad pairs formed by two vertices $v\in T_i$ and $u\in T_j$ such that $i\neq j$. That is,  $\bpn(\Gamma) \le \sum_{i=1}^l k\Delta n_i = k \Delta n$, finishing the proof.
\end{proof}

\section{Conclusion and Open Problems}\label{se:conclusion}

In this paper we initiated the study of computing planar drawings of graphs for which the number of violations of a given proximity rule is as small as possible. The research is motivated by the observation that only restricted families of graphs admit a proximity drawing. As a notable case of study we choose trees and their realizations as Euclidean minimum spanning trees of their vertex set, which we call EMST-drawings in the paper. We showed that, in contrast with the NP-hardness result in~\cite{EadesW1996algo}, deciding when a caterpillar is EMST-drawable is linear-time solvable. Also, we can compute in linear time a planar drawing of a degree-six caterpillar with the minimum number of bad pairs. For general trees of degree larger than 6, we established upper bounds on their bad pair number and computed the exact bad pair number of stars. We conclude by listing three open problems that naturally stem from our work.

\begin{enumerate}
\item Study the area requirement of EMST-drawable caterpillars. The algorithm in the proof of \cref{thm:char-6-EMST} computes an EMST-drawing with exponential area. The proof by Angelini et al.~\cite{DBLP:journals/comgeo/AngeliniBCFKS14} about the area requirement of EMST-drawings of trees does not apply to caterpillars. 
\item How hard is it to decide whether a tree of bounded pathwidth is EMST-drawable?
\item Study the bad pair number of graphs also for other definitions of proximity drawings.  For example, a \emph{Gabriel bad pair} $u,v$ is either a pair of non-adjacent vertices such that the disk with $u$ and $v$ as antipodal points does not contain any other vertex or a pair of adjacent vertices for which this disk does contain some other vertex. We would like to compute planar drawings of trees with the minimum number of Gabriel bad pairs. We recall that the trees that admit a planar drawing without Gabriel bad pairs are a subset of those having maximum vertex degree four~\cite{DBLP:journals/algorithmica/BoseLL96}.
\end{enumerate}

\bibliographystyle{plainurl}
\bibliography{references} 

\clearpage

\appendix

\section{Omitted Parts From the Proof of \cref{thm:char-6-EMST}}\label{app:MST-caterpillars}

\claimFiveSixCater* \label{claim:56cater*}

\begin{claimproof}
Without loss of generality, we assume that $d_t=5$, as the other case is then directly implied. 
By \cref{cor:644grid}, any drawing of the caterpillar $T'= (6,4,\dots, 4)$ obtained by removing a leaf adjacent to $v_t$ is an embedding into a triangular grid.
Now, if $T$ was EMST-drawable, we would be able to add the leaf that we removed to the realization of $T'$ and obtain a realization of $T$. But, by \cref{lem:64grid}, this is impossible, since the leaves adjacent to $v_t$ need to coincide with the vertices of the regular hexagon centered at $v_t$ and all such vertices are already occupied by the vertices of $T'$. 
\end{claimproof}

\claimSixThreeSixCater* \label{cl:636cater*}

\begin{claimproof}
    We first prove the claim for the case $t=3$ (and thus $i=2)$.  Since $v_1$ and $v_3$ are vertices of degree $6$, by \cref{lem:6uniq} there is a unique way to draw $K_{1,6}$ as an EMST. Consequently, any EMST-drawing of the caterpillar $(6,2,6)$ consists of merging of two such drawings of $K_{1,6}$ as depicted in \cref{fig:636caterpillar}. In particular, if one of the two copies of $K_{1,6}$ was drawn with longer edges, or if the spine vertices were not collinear, this would result in a drawing that is not an EMST. 
    
    Suppose for a contradiction, that it is possible to add a leaf vertex $x$ adjacent to $v_2$ into the EMST-drawing. 
    Consider a tree $T'$ with an additional vertex $x'$ attached to $v_2$. Let us draw the edge connecting $v_2$ and $x'$ on the line defined through $v_2$ and $x$ and let it have the same length as edge  $v_2x$, as in \cref{fig:636caterpillar}. Since $x$ does not belong to any bad pair, neither does $x'$, and we have obtained a drawing of $T'=(6,4,6)$. Note that this drawing is not the triangular grid drawing, as $x$ and $x'$ would then coincide with leaves attached to $v_1$ and $v_3$, and thus contradicts \cref{cor:644grid}.

    Now assume that $t\ge 4$. Then, we have that $(d_1,d_2,\dots, d_{i-1}) = (6,4,\dots,4)$ and $(d_{i+1}, d_{i+2},\dots d_t) = (4,4,\dots,6)$, and in particular $v_i$ is a leaf of spinars $T_{1(i-1)}$ and $T_{(i+1)t}$. By \cref{cor:644grid}, both any EMST-drawing of $T_{1(i-1)}$ and $T_{(i+1)t}$ is an embedding into a triangular grid. Hence any drawing of the caterpillar $T'$ obtained from $T$ by removing the leaf adjacent to $v_i$ is obtained by gluing the drawings of $T_{1(i-1)}$ and $T_{(i+1)t}$ at $v_i$. Then, all the vertices of a regular hexagon centered at $v_i$ are occupied by vertices of the drawing of $T'$. Now, by similar arguments as before, it is impossible to complete the drawing of $T'$ to a drawing of $T$.
\end{claimproof}
    \begin{figure}[h]
        \centering
        \includegraphics[page=4]{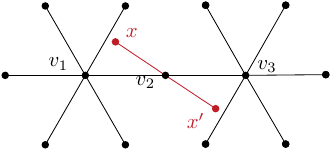}
        \caption{Drawing of the caterpillar $(6,3,6)$, and its impossible extension to a drawing of $(6,4,6)$.}
        \label{fig:636caterpillar}
    \end{figure}

\claimCidrawing* \label{cl:S_idrawing*}
\begin{claimproof}
    We first consider the case $i\in \{1,\dots, l-1\}$, i.e  an "inner" subcaterpillar $S_i=(d_1,\dots,d_s)$.  Note that, since the original caterpillar $T$ satisfies the conditions of the theorem,  so does every $S_i$. Additionally, $d_1, d_s \in \{1,2,3\}$, as otherwise the corresponding $v_1$ or $v_s$ in $T$ would be of degree $5$ or $6$, which would require existence of a vertex of degree at most 3 between $v_1$ or $v_s$ and its spine neighbor in $F_i$ or $F_{i+1}$, respectively, because these are of degree 5.

    If there are no vertices of degree $6$ in $S_i$, then $S_i$ is a subgraph of a $(3,4,\dots,4,3)$ caterpillar for which the statement is clearly true; see \cref{fig:34443}. Hence, we assume that $S_i$ contains at least one vertex of degree $6$.
    \begin{figure}[h]
        \centering
        \includegraphics[page=9,scale=0.8]{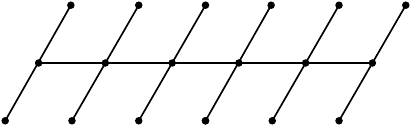}
        \caption{Drawing of a $(3,4,\dots,4,3)$ caterpillar satisfying the conditions of \cref{cl:S_idrawing}.}
        \label{fig:34443}
    \end{figure}

    Let $j_{\min},j_{\max}$ be the minimum and maximum indices such that $d_{j_{\min}}=d_{j_{\max}} = 6$. Then there is at least one $j_1<j_{\min}$  and at least one $j_2 > j_{\max}$ such that $d_{j_1},d_{j_2}\in\{2,3\}$. 
    Now, we construct a caterpillar $S_i'$ by adding leaves to $S_i$ as long as we can, while still preserving the conditions described above. In particular, between every two consecutive vertices of degree $6$ on the spine of $S_i'$, there are either exactly two vertices of degree $3$ or exactly one vertex of degree $2$, and the remaining vertices have degree $4$. Similarly, the two segments of the spine coming before $v_{j_{\min}}$ and after $v_{j_{\max}}$ contain a single vertex of degree $3$ and all other vertices have degree $4$. It is easy to see that we can obtain a EMST-drawing of $S_i$ from a realization of $S_i'$ by removing leaves.

    Now we describe how to draw $S_i'$, see \cref{fig:T_idrawing}. 
    The first row in the figure shows how to draw the initial  and final segments of $S_i'$ (from $v_1$ and $v_{j_{\min}}$  and from $v_{j_{\max}}$ to $v_{s}$), note that we can add as many vertices of degree $4$ before or after $v_{j_1}$ to the drawing. The second row of the figure shows how to draw the part of $S_i'$ between two consecutive vertices of degree $6$, with two distinct cases based on whether between them there is a single vertex of degree $2$ or there are two vertices of degree $3$. 
    In the last row of the figure, we show how these drawings are combined together. 

    In the case where $i=0$ ($l$), the initial (final) segment of the spine may contain only vertices of degree $4$, so we first remove a leaf adjacent to $v_1$ ($v_s$), draw the remaining tree as described above, and add the removed leaf back so that the resulting drawing is still an embedding into the triangular grid.

\end{claimproof}

\claimFFFiveCater* \label{cl:555caterpillars*}
\begin{claimproof}
    We draw $T$ by using \cref{alg:monma-suri}, hence it is clear that we obtain a strong EMST-drawing of $T$. 
    First we draw $v_1$ so that the vertices adjacent to it form a regular pentagon. We can assume (up to rotation) that $v_2$ is the neighbor of $v_1$ with the largest $x$-coordinate and that the edge $v_1v_2$ is drawn horizontally. Now, let $i\in\{2,\dots,t\}$ be arbitrary and assume that we have drawn the spinar

    $T_{1i}$ according to \cref{alg:monma-suri}. Let $v_{i-1}=v_i^{0},v_i^1,v_i^2,v_i^3,v_{i}^4$ be the neighbors of $v_i$ in the current drawing, ordered clockwise as seen from $v_i$. If the edge $v_{i-1}v_i$ has positive slope, we set $v_{i+1} = v_i^4$ and if $v_{i-1}v_i$ has negative slope, we set $v_{i+1}= v_{i}^3$. Therefore the slope of the edge $v_{i}v_{i+1}$ is exactly the opposite from the slope of the edge $v_{i-1}v_i$, and since this holds for every $i\ge 2$, the spine is an $x$-monotone curve. 
\end{claimproof}

\section{Proof of \cref{prop:disjointsubcaterpillars}}

\propdisjointsubcaterpillars* \label{prop:disjointsubcaterpillars*}
\begin{proof}
Suppose $S_1=\widetilde{T}_{ih}$ and $S_2=\widetilde{T}_{kl}$ with $i\le h\le k\le l$.
\begin{enumerate}
    \item If $S_1$ and $S_2$ do not share spine vertices, we have $h<k$. In this case, $S_1'$ and $S_2'$ are either disjoint, or their intersection is $\{v_{i_h},v_{i_k}\}$  (the latter happening when $k=h+1$ and $i_k=i_h+1$, i.e., when $v_{i_k}$ and $v_{i_k}$ are consecutive vertices on the spine of $T$). The only potential bad pair that belongs to both $S_1'$ and $S_2'$ in $\Gamma$ is $v_{i_h}$ and $ v_{i_k}$, but these vertices are adjacent and thus not a bad pair.
    \item 
    If $S_1$ and $S_2$  share exactly one spine vertex, we have $h=k$. In this case, the common vertices of $S_1'$ and $S_2'$ are $v_{i_h-1},v_{i_h},v_{i_h+1}$ and the leaves adjacent to $v_{i_h}$. Among them, $v_{i_h}$ is not involved in any bad pair, and the only bad pair that does not involve a leaf adjacent to $v_{i_h}$ can be the pair $v_{i_h-1}$ and $v_{i_h+1}$.
\end{enumerate}
See \cref{fig:subcaterpillars} for illustrations of both cases.
\end{proof}

\section{Omitted Parts From the Proof of \cref{thm:bpn-caterpillarmax6}}\label{app:bpn-caterpillarmax6}

\claimCatmaxSixA* \label{cl:bpn-catmax6-1*}

\begin{claimproof}
    The assumptions imply that $T=(6,4,\dots,4,6)$ and both vertices of $\widetilde{T}$ are red. \cref{thm:char-6-EMST} implies that $\bpn(T)\ge 1$. Suppose $\Gamma$ is a drawing of $T$ with just one bad pair. If this bad pair involves any leaf of the caterpillar, delete such a leaf. The resulting drawing has no bad pairs, so we have an EMST-drawing of $(5,4,\dots,4,6)$ or of $(6,4,\dots,4,3,4,\dots,4,6)$, a contradiction with \cref{thm:char-6-EMST}. 
    If the bad pair does not involve any leaf, 
    i.e. it is formed by two spine vertices, $v_i$ and $v_{i'}$, consider the drawings of $T_{1(i'-1)}$ and $T_{(i+1)t}$. By \cref{lem:64grid}, both of these drawings are on a triangular grid, which leads to a contradiction as there occurs a pair of leaves drawn at the same point.  
    
\end{claimproof}

\claimCatmaxSixB* \label{cl:bpn-catmax6-2*}

\begin{claimproof}
    In this case, all three vertices of $\widetilde{T}$ are red. Since $T=(6,4,\dots,4,6,4,\dots,4,6)$ contains $(6,4,\dots,4,6)$ as an induced subgraph, we know from \cref{cl:bpn-catmax6-1} that $\bpn(T)\ge 2$. Suppose for contradiction, that $\bpn(T)= 2$. Consider a drawing $\Gamma$ with exactly two bad pairs.
    The spinar $T_{1,i_2}$ on $v_1,\dots,v_{i_2}$ is $(6,4,\dots,4,6)$, and hence both bad pairs of $T$ are bad pairs of $T_{1,i_2}$. By symmetry, they both appear in $T_{i_2,t}$. 
    Part \ref{i2} of \cref{prop:disjointsubcaterpillars} implies that both bad pairs appear in the neighborhood of $v_{i_2}$ and, since at most one of them involves only spine vertices, at least one bad pair involves a leaf adjacent to $v_{i_2}$.
   
    Delete this leaf. The resulting drawing of $(6,4,\dots,4,5,4,\dots,4,6)$ has at most one bad pair, but according to \cref{cl:bpn-catmax6-1}, its induced subgraph $(6,4,\dots,4,6)$ in contrary requires at least two bad pairs. 
    Hence, $\bpn(T)\ge 3$.
\end{claimproof}

\claimCatmaxSixC* \label{cl:bpn-catmax6-3*}
\begin{claimproof}
    In this case, all $s$ vertices of $\widetilde{T}$ are red. If $s$ is even, then the spine of $\widetilde{T}$ contains a matching of size $\frac{s}{2}$, and hence $T$ contains $\frac{s}{2}$ spinars of isomorphism type $(6,4,\dots,4,6)$. By \cref{prop:disjointsubcaterpillars}, part \ref{i1}, these spinars do not share bad pairs. By \cref{cl:bpn-catmax6-1}, each of them produces at least two bad pairs in any drawing, so in total we have at least $2\cdot \frac{s}{2}=s$ bad pairs in any drawing of $T$.
    If $s\ge 3$ is odd, the spine of $\widetilde{T}$ can be partitioned into a path of length 2 and a matching of size $\frac{s-3}{2}$. Hence $T$ consists of a spinar $(6,4,\dots,4,6,4,\dots,4,6)$ and $\frac{s-3}{2}$ spinars of the isomorphism type  $(6,4,\dots,4,6)$ with disjoint spines.  By \cref{prop:disjointsubcaterpillars,cl:bpn-catmax6-1,cl:bpn-catmax6-2}, they produce at least $3+2\cdot\frac{s-3}{2}=s$ bad pairs in any drawing of $T$.
\end{claimproof}

\claimCatmaxSixD* \label{cl:bpn-catmax6-4*}
\begin{claimproof}
    In this case, $T=(5,4,4,\ldots,4,6)$ and \cref{thm:char-6-EMST} says that any drawing of $T$ has at least one bad pair.
\end{claimproof}

\claimCatmaxSixE* \label{cl:bpn-catmax6-5*}

\begin{claimproof}
    In this case, the two vertices of degree 6 of $\widetilde{T}$ are red. The isomorphism type of $T$ is  $(6,4,\dots,4,5,4,\dots,4,6)$, and thus $T$ contains a caterpillar of type $(6,4,\dots,4,6)$ as an induced subgraph. By \cref{cl:bpn-catmax6-2}, already this subgraph requires at least two bad pairs.
\end{claimproof}

\claimCatmaxSixF* \label{cl:bpn-catmax6-6*}

\begin{claimproof}
    Note that if $\widetilde{T}$ is a dark purple path, its edges are oriented in an alternating way (because vertices of degrees 5 and 6 alternate on this path). Let $T^*$ be the caterpillar obtained by deleting one leaf from every vertex of degree 5. Then $\widetilde{T^*}$ contains a red path with $\rv(T)$ red vertices, and hence, by~\cref{cl:bpn-catmax6-3}, every drawing of $T^*$ contains at least $\rv(T)$ bad pairs, and so does any drawing of $T$. 
\end{claimproof}

\claimCatmaxSixG* \label{cl:bpn-catmax6-7*}

\begin{claimproof}
    In this case, the spine of $\widetilde{T}$ contains a matching of size $\floor{\frac{h+1}{2}}$, each spinar of $T$ corresponding to a dark blue edge of this matching produces at least one bad pair by \cref{thm:char-6-EMST}, and all these bad pairs are distinct by \cref{prop:disjointsubcaterpillars}, part \ref{i1}.
\end{claimproof}

\section{Proof of \cref{obs:star_pentagonal}}

\obsStarPentagonal* \label{obs:star_pentagonal*}

\begin{proof}
    The proof goes by induction on $s$.
    For $s=1$, the statement is obvious.
    For $s\ge2$, we construct a pentagonal solution for $s$ from a pentagonal solution for $s-1$.
    We add the new vertex to the smallest group in the $5$-partition given by the pentagonal solution for $s-1$.
    The number of bad pairs in the resulting solution is:
    \[
        t'(s-1,5) + \floor{\frac{s-1}{5}} = t'(s,5)
        \text{.} \qedhere
    \]
\end{proof}

\section{Proof of \cref{cor:boundeddepth}}
\corBoundeddepth* \label{cor:boundeddepth*}
\begin{proof}
    We closely follow the approach used in the proof of~\cref{thm:max-D-ary-tree}, using the same construction for $\Gamma(T)$ and the same notation. We prove the claim by induction on $h$. For $h=1$, the claim is trivially true.

    Let $T$ be a tree of height $h$, and let $T_i$, $n_i$, $\Gamma(T)$ be as in the proof of~\cref{thm:max-D-ary-tree}. We have
    \begin{align*}
    \bpn(T)&\leq \sum_{i=1}^k \bpn(T_i)+\sum_{i=1}^k (k-i)n_i\leq \sum_{i=1}^k (h-1)\Delta n_i+\sum \Delta n_i\\
    &\leq (h-1)\Delta n+\Delta n=h\Delta n\text{.}\qedhere
    \end{align*}
\end{proof}

\end{document}